\documentclass[final,5p,times,twocolumn]{elsarticle}%
\journal{Signal Processing}

\usepackage{amsmath,amssymb,graphicx,epsfig,algorithm,algorithmic,epstopdf, url}
\usepackage[utf8]{inputenc}
\usepackage[mathscr]{euscript}
\usepackage{color}
\usepackage{xcolor}
\usepackage{amsthm}
\usepackage{bbm}
\usepackage{booktabs}
\usepackage{soul}
\usepackage{balance}
\usepackage{algorithm}
\usepackage{algorithmic}

\usepackage{siunitx}
\def\minwrt[#1]{\underset{#1}{\text{minimize }}}
\def\argminwrt[#1]{\underset{#1}{\text{arg min }}}
\def\maxwrt[#1]{\underset{#1}{\text{maximize }}}
\def\maxemphwrt[#1]{\underset{#1}{\text{\emph{maximize} }}}

\newtheorem{theorem}{Theorem}
\newtheorem{remark}{Remark}
\newtheorem{proposition}{Proposition}
\newtheorem{lemma}{Lemma}

\newtheorem{corollary}{Corollary}

\newcommand{\norm}[1]{\left\lVert#1\right\rVert}
\newcommand{\abs}[1]{\left|#1\right|}

\def\RC{{\mathbb{C}}}
\def\RN{{\mathbb{N}}}
\def\RR{{\mathbb{R}}}

\newcommand{\cov}{r}
\newcommand{\covest}{\hat{r}}

\newcommand{\asympvarcross}{\Psi_T'}
\newcommand{\asympvarautoamp}{\Psi_T}

\newcommand{\fouriervec}{a}

\newcommand{\randamp}{z}

\newcommand{\freqband}{\mathcal{I}_B}

\newcommand{\var}[1]{\mathrm{Var}\left(#1\right)}

\newcommand{\expop}{\mathbb{E}}
\newcommand{\expect}[1]{\expop\left(#1\right)}

\newcommand{\xamp}{\alpha}
\newcommand{\yamp}{\beta}

\newcommand{\freq}{\theta}

\newcommand{\kernel}{f}

\newcommand{\errorfunc}{\epsilon}

\newcommand{\asympratio}{\gamma}

\newcommand{\weakstarconv}{\stackrel{\ast}{\rightharpoonup}}

\newcommand{\kurtrat}{\kappa}

\newcommand{\jk}[1]{{\color{blue}{#1}}}

        \makeatletter
        \def\fps@eqnfloat{!t}
        \def\ftype@eqnfloat{4}
        
        \newenvironment{eqnfloat*}
               {\@dblfloat{eqnfloat}}
               {\end@dblfloat}
        \makeatother
\renewenvironment{thebibliography}[1]{%
 \begin{oldthebibliography}{#1}%
 \setlength{\parskip}{0ex}%
 \setlength{\itemsep}{0ex}%
}%
{%
\end{oldthebibliography}%
}%

\begin{document}
\begin{frontmatter}


\title{
Mixed-Spectrum Signals -- Discrete Approximations and \\Variance Expressions for Covariance Estimates
}

\tnotetext[t1]{This work was supported in part by the Swedish Research Council grant 2020-03454.}
\author[KUL]{Filip Elvander\corref{cor1}} 
\ead{firstname.lastname@esat.kuleuven.be}

\author[KTH]{Johan Karlsson}
\ead{firstname.lastname@math.kth.se}

\address[KUL]{Stadius Center for Dynamical Systems, Signal Processing and Data Analytics, KU Leuven, Leuven, Belgium}
\address[KTH]{Department of Mathematics, KTH Royal Institute of Technology, Stockholm, Sweden}

\begin{abstract}
The estimation of the covariance function of a stochastic process, or signal, is of integral importance for a multitude of signal processing applications. In this work, we derive closed-form expressions for the variance of covariance estimates for mixed-spectrum signals, i.e., spectra containing both absolutely continuous and singular parts. The results cover both finite-sample and asymptotic regimes, allowing for assessing the exact speed of convergence of estimates to their expectations, as well as their limiting behavior. As is shown, such covariance estimates may converge even for non-ergodic processes. Furthermore, we consider approximating signals with arbitrary spectral densities by sequences of singular spectrum, i.e., sinusoidal, processes, and derive the limiting behavior of covariance estimates as both the sample size and the number of sinusoidal components tend to infinity. We show that the asymptotic regime variance can be described by a time-frequency resolution product, with dramatically different behavior depending on how the sinusoidal approximation is constructed. In a few numerical examples we illustrate the theory and the corresponding  implications for direction of arrival estimation. 

\end{abstract}
\end{frontmatter}

\section{Introduction}
Modeling signals that impinge on sensor arrays appear in a large variety of signal processing applications, including radar, sonar, and audio signal processing \cite{KrimV96,VanTrees02,GannotVMGO17_25}.
Commonly in such applications, one seeks a spatial spectrum, describing the distribution of signal energy over the space of interest, e.g., azimuth and elevation in direction of arrival (DoA) estimation \cite{TrinhVP20_68}, allowing for localizing and tracking targets \cite{ElvanderJK18_66,ElvanderHJK20_171} or for performing spatial filtering of the sensor signals \cite{AliWM19_27}. In practice, the spatial spectrum is often inferred from the array covariance matrix as in, e.g., optimal filtering such as the Capon method \cite{Capon69}, subspace methods as ESPRIT and MUSIC \cite{PaulrajKRasilomar85,Schmidt79}, as well as more recent contributions exploiting sparse representations \cite{StoicaBL11_59b} as well as knowledge of underlying dynamics \cite{ElvanderHJK20_171}. 

Commonly, it is assumed that the impinging signals are narrowband, or that they may be decomposed into narrowband components through filtering or by the use of short time Fourier transforms \cite{Bohme86_10}, and that time delays may be described as phase shifts of the source signal waveform \cite{StoicaM05}. 
Assuming that the impinging signals are uncorrelated, this then induces a low-rank structure in the array covariance matrix, which is exploited in estimation of the spatial spectrum, e.g., using the Caratheodory-Fejer theorem for Toeplitz matrices in the case of uniform linear arrays \cite{GrenanderS58,ElvanderJK18_66}. The success of covariance based approaches are thus dependent on the availability of accurate estimates of the array covariance matrix, and, in particular, on the speed of convergence of finite-sample estimators to their expectation. 
Typically it is assumed that a large number of independent samples are available for estimating the covariance \cite{Jaffer88,StoicaN89_37,StoicaN95_14}. However, the narrowband assumption would imply that the samples are highly correlated also over large time horizons. In particular, perfectly narrowband signals, i.e., signals whose spectra have support of measure zero, are not ergodic and exhibit no mixing. In practice, the signals may be band-limited but with non-zero bandwidth, and a relevant question is then how the spectral properties affect the accuracy of the covariance estimates for this class of signals. 
For signals decomposable as a finite sum of fixed magnitude sinusoids and a moving average process, the asymptotic normality of the array sample covariance matrix was proved in \cite{Delmas01_47}, with the asymptotic performance of frequency estimation algorithms being presented in \cite{Delmas02_50,DelmasM03_83}.
However, to the best of the authors' knowledge, no finite-sample results for the accuracy of covariance estimates for signals with general spectra exist in the signal processing literature. In particular, there are no widely available results on the dependence of finite-sample accuracy on the impinging signals' spectra.

In this work, we consider the problem of covariance estimation for signals with mixed spectra, consisting of a Gaussian part with a spectral density and a singular part. Specifically, we derive closed form expressions for the finite-sample variance of the covariance estimates, allowing for exactly quantifying the speed of convergence for estimates to their expected values. As is shown, the properties of the covariance estimates vary considerably depending on the distribution, and specifically its kurtosis, of the random amplitudes of the sinusoidal components used for modeling the singular parts of the spectrum.
In particular, we show that a model with fixed magnitudes and random phases is the only model with circular symmetric components that yields statistically consistent covariance estimates.
Furthermore, for Gaussian processes with arbitrary spectral densities, we consider utilizing sinusoidal components for constructing singular spectrum approximations. For a general class of amplitudes with circularly symmetric distributions, we show that these approximations converge in distribution to the target process as the spacing of frequency grid goes to zero.
When estimating covariances from such a signal, a relevant limit is when both the time interval and the number of sinusoidal component in the approximation tends to infinity. We give an explicit expression for the asymptotic variance for this case and note that the variance is a function of the product of the time window and the frequency resolution.
As we show, depending on the kurtosis of the distribution of the component amplitudes, the covariance estimates for such singular approximations have variances that either upper or lower bound that of the target process.
In particular, this implies that evaluating, e.g., direction of arrival estimators that make use of second-order moments based on data simulated from such sinusoidal approximations may yield results that are not representative for such estimators' performance on data generated by processes with spectral densities. We derive conditions for when such singular approximations perfectly mimic the target process, thereby allowing for implementing any array processing scenario with arbitrary spectra.
\section{Signal model}
Consider two scalar wide-sense stationary (WSS) zero-mean complex circularly symmetric stochastic processes $x$ and $y$ on the real line with power spectra $d\mu_x$ and $d\mu_y$, and cross spectrum $d\mu_{xy}$. The covariance functions are the Fourier transforms of the spectra, and thus given by
\begin{equation} \label{eq:cov_funcs}
\begin{aligned}
	\cov_x(\tau)&\triangleq\expop\left( x(t)\overline{x(t-\tau)} \right) = \int_{-\infty}^\infty e^{i2\pi\freq \tau}d\mu_x(\freq),\\ 
\cov_y(\tau) &\triangleq \expop\left( y(t)\overline{y(t-\tau)} \right) =  \int_{-\infty}^\infty e^{i2\pi\freq \tau}d\mu_y(\freq),\\
\cov_{xy}(\tau) &\triangleq\expop\left( x(t)\overline{y(t-\tau)} \right)		 = \int_{-\infty}^\infty e^{i2\pi\freq \tau}d\mu_{xy}(\freq),
\end{aligned}
\end{equation}
for $\tau \in \RR$, where  $\overline{z}$ denotes the complex conjugate of a complex scalar $z$, and $\expop(\cdot)$ denotes the expectation operator.
In this paper we focus on band-limited signals, and we will assume that all power spectra are 
supported in the frequency 
interval $\mathcal{I}_B = [\freq_c - B/2,\freq_c+B/2]$ and
given by
\begin{subequations}\label{eq:spectra}
\begin{align}
d\mu_x(\freq) &= \Phi_x(\freq)d\freq + \sum_{k=1}^{N_x} \xamp_k^2\delta_{\freq_k^x}(\freq), \\
d\mu_y(\freq) &= \Phi_y(\freq)d\freq + \sum_{k=1}^{N_y} \yamp_k^2\delta_{\freq_k^y}(\freq),\\ 
d\mu_{xy}(\freq) &= \Phi_{xy}(\freq)d\freq.
\end{align}
\end{subequations}
Here $\Phi_x$, $\Phi_y$, and $\Phi_{xy}$ are densities in $L_1(\mathcal{I}_B)$, $\xamp_k$ and $\yamp_k$ are positive constants, and  $\delta_{\freq_k}(\freq) \triangleq \delta(\freq-\freq_k)$ where $\delta$ denotes the Dirac delta function.%
\footnote{This implies that the singular parts of $x$ and $y$ are uncorrelated.} 
Note that $\freq_c$ is the center frequency and $B$ is the bandwidth of the signals. 

Herein, we will consider the following general model for $x$ and $y$ consistent with \eqref{eq:spectra}:
\begin{subequations} \label{eq:general_model}
\begin{align}
x(t)=x_a(t) +x_s(t)= x_a(t)+\sum_{k=1}^{N_x}\randamp_k^{(x)} e^{i 2\pi\freq_k^x t}\\
y(t)=y_a(t) +y_s(t)= y_a(t)+\sum_{\jk{\ell}=1}^{N_y}\randamp_\ell^{(y)} e^{i 2\pi\freq_\ell^y t}
\end{align}
\end{subequations}
for $t\in \RR$, where $x_a$ and $y_a$ are Gaussian, band-limited, zero-mean random processes with absolutely continuous (cross)spectra $\Phi_x$, $\Phi_y$, $\Phi_{xy}$, and where $\randamp_k^{(x)}$ and $\randamp_\ell^{(y)}$ are independent, zero-mean, circularly symmetric complex random variables with finite fourth absolute moments such that
\begin{align*}
	\expect{|\randamp_k^{(x)}|^2} = \xamp_k^2 \;,\; \kurtrat_k^{(x)} \triangleq \frac{\expect{|\randamp_k^{(x)}|^4}}{\expect{|\randamp_k^{(x)}|^2}^2} < \infty,
\end{align*}
and analogously for $\randamp_\ell^{(y)}$, where $\kurtrat_k^{(x)}$ and $\kurtrat_\ell^{(y)}$ are the kurtosis parameters. Although all results presented in this paper hold for any instance of the model in \eqref{eq:general_model}, we will throughout give particular attention to two particular models displaying some unique properties in the context of approximation and covariance estimation. Specifically, we consider
\begin{align} \label{eq:random_amplitude_model}
	\randamp_k^{(x)} \sim \mathcal{CN}(0,\xamp_k^2)\;,\; \randamp_\ell^{(y)} \sim \mathcal{CN}(0,\yamp_\ell^2),
\end{align}
%
%
\begin{figure}[t]
        \centering
            \includegraphics[width=.46\textwidth]{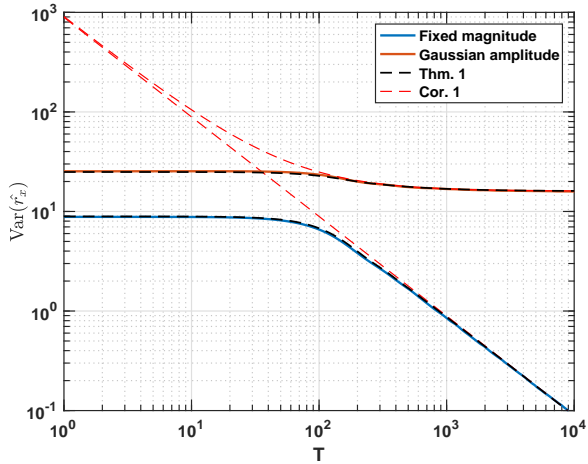}
           \caption{Empirical variance of auto-covariance estimate as a function of the measurement duration $T$ for a band-limited processes $x = x_a + x_s$, generated according to \eqref{eq:random_amplitude_model} as well as according to \eqref{eq:random_phase_model}. Also presented are finite-sample as well as asymptotic large-sample theoretical values of $\var{\covest_{x}(\tau;T)}$ according to Theorem~\ref{thm:auto_cov_general} and Corollary~\ref{cor:auto_cov_general}, respectively.} 
            \label{fig:example_autocovariance}
\vspace{-2mm}\end{figure}
%
%
%
%
%
i.e., the amplitudes are complex Gaussian random variables, as well as a fixed magnitude but random phase model
\begin{align} \label{eq:random_phase_model}
	\randamp_k^{(x)} = \xamp_k e^{i\varphi_k^{(x)}}\;,\; \randamp_\ell^{(y)} = \yamp_\ell e^{i\varphi_\ell^{(y)}},
\end{align}
where $\varphi_k^x$ are random variables with uniform distribution $U((-\pi,\pi])$ (cf. \cite{DelmasM06_86}). It may here be noted that the kurtosis parameters for these two models are $\kurtrat_k = 2$ for \eqref{eq:random_amplitude_model} and $\kurtrat_k = 1$ for \eqref{eq:random_phase_model}. For the covariance functions in \eqref{eq:cov_funcs}, consider the standard estimates
\begin{subequations}\label{eq:autocov_est}
\begin{align}
	\covest_{x}(\tau;T) &= \frac{1}{T}\int_{t=0}^T x(t)\overline{x(t-\tau)}dt ,\label{eq:autocov_est_a}\\
		\covest_{y}(\tau;T) &= \frac{1}{T}\int_{t=0}^T y(t)\overline{y(t-\tau)}dt ,\\
	\covest_{xy}(\tau;T) &= \frac{1}{T}\int_{t=0}^T x(t)\overline{y(t-\tau)}dt ,
	\end{align}
\end{subequations}
where $T$ is the averaging time. As we will see, although having the same spectra and covariance functions, processes constructed according to \eqref{eq:general_model} display considerable differences when it comes to estimating the covariances \eqref{eq:autocov_est} depending on how the distributions of the amplitudes $\randamp_k^{(x)}$ and $\randamp_\ell^{(y)}$ are chosen. In particular, we will show that the convergence, as $T \to \infty$, of the covariance estimates to their respective expectations depend on the structures of $d\mu_x$ and $d\mu_y$, as well as on kurtosis of the distribution of the amplitudes. Specifically, we are interested in under what conditions the estimators in \eqref{eq:autocov_est} are consistent estimators. A motivating example utilizing the models in \eqref{eq:random_amplitude_model} and \eqref{eq:random_phase_model} illustrating these problems is presented in the next section.
%
\begin{figure}[t]
        \centering
            \includegraphics[width=.46\textwidth]{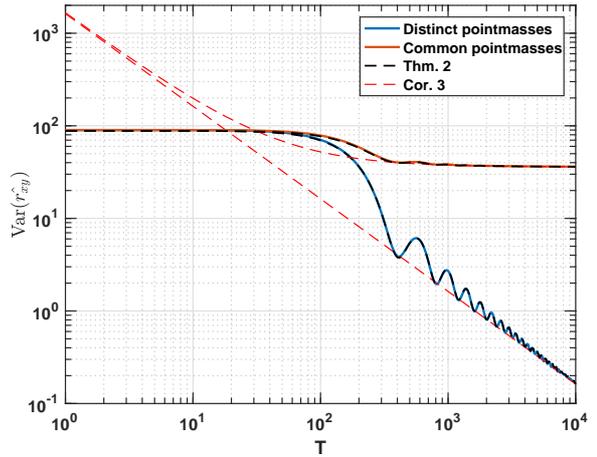}
           \caption{Empirical variance of cross-covariance estimate as a function of the measurement duration $T$ for two band-limited processes $x = x_a + x_s$ and $y = y_s$, generated according to \eqref{eq:random_phase_model}, with both distinct and common point masses. Also presented are finite-sample as well as asymptotic large-sample theoretical values of $\var{\covest_{xy}(\tau;T)}$ according to Theorem~\ref{thm:cross_cov} and Corollary~\ref{cor:asymptote}, respectively.} 
            \label{fig:example_crosscovariance}
\vspace{-2mm}\end{figure}
%
%
\section{Motivating examples}
Consider an array processing scenario in which two sources, emitting the signals $x$ and $y$ respectively, impinge on a set of sensors. Considering two of the sensors, the measured signals, $s_1$ and $s_2$, are given by\footnote{For simplicity of the exposition, but without loss of generality for the discussion, we here assume lossless propagation.}
\begin{align*}
	s_1(t) &= x\left(t-\tau_x^{(1)}\right) + y\left(t-\tau_y^{(1)}\right) \\
	s_2(t) &= x\left(t-\tau_x^{(2)}\right) + y\left(t-\tau_y^{(2)}\right),
\end{align*}
where $\tau_x^{(1)},\tau_x^{(2)},\tau_y^{(1)},\tau_y^{(2)}$ are time delays determined by the distance between the sources and the sensors. Then, in order to localize the signal sources \cite{AdlerW19_eusipco} or perform noise reduction \cite{ElvanderAJW19_eusipco}, one typically considers estimates of the array cross-covariance in order to, e.g., fit parametric models \cite{OtterstenSR98_8}. Clearly, the success of such approaches depends on the convergence of empirical moments to their theoretical counterparts, which is determined by the convergence of $\covest_x, \covest_y$, and $\covest_{xy}$ to their respective expectations.
Note that for these problems a continuous-time model of the signal is required in many cases. This is since in array signal processing, the angle or location of signal sources are continuous variables and it is thus not enough to restrict time delays $\tau_x^{(1)}, \tau_x^{(2)}, \tau_y^{(1)}, \tau_y^{(2)}$, and thereby covariance lags, to a discrete grid. Furthermore, unless the signals are perfectly narrowband, time delays cannot simply be modelled as phase shifts of the signal waveforms. 

In this setting, consider estimating the auto-covariance, $\cov_x(\tau)$, of a band-limited stochastic processes $x = x_a + x_s$ realized by the model in \eqref{eq:general_model}, where $x_s$ is generated according to \eqref{eq:random_amplitude_model} or \eqref{eq:random_phase_model}, and where $x_a$ has a flat spectral density with bandwidth $B = 10^{-2}$. Here, $x_s$ consists of a single point mass in located in the same band as the spectrum of $x_a$.
Figure~\ref{fig:example_autocovariance} displays the empirical variance, obtained in a Monte Carlo simulation study, of the standard auto-covariance estimate $\covest_{x}(\tau;T)$ as a function of the averaging time $T$. As can be seen, for the fixed magnitude model \eqref{eq:random_phase_model}, the variance tends to zero, whereas it for the Gaussian amplitude model \eqref{eq:random_amplitude_model} converges to a strictly positive number.

Next, introduce a second stochastic process $y = y_s$, independent of $x$, consisting of two sinusoidal components, and thus the spectrum consists of two point masses. Let the frequencies of the point masses belong to the same band as the spectrum of $x_a$, and consider two scenarios. Firstly, when both point masses are distinct from the point mass of $x_s$, and secondly, when one of the point masses of $y_s$ is located at the frequency corresponding to  the point mass of $x_s$. 
Figure~\ref{fig:example_crosscovariance} shows the empirical variance of the cross-covariance estimate $\covest_{xy}(\tau;T)$ for these two scenarios.
As can be seen, for the first scenario with no common point mass, the estimator variance tends to zero as $T$ increases, whereas for the second scenario with a common point mass, the variance remains bounded away from zero.
It may be noted that Figures~\ref{fig:example_autocovariance} and \ref{fig:example_crosscovariance} in addition to the empirical variances also show theoretically computed finite-sample and asymptotic values of the estimator variance. In the following section, we derive closed-form expressions for these quantities, explaining the observations from the examples.
\section{Estimating the covariance function}
As seen in the previous section, the variance of the cross-covariance estimate, i.e., $\text{Var}\left( \covest_{xy}(\tau;T) \right)$, behaves dramatically different depending on whether or not the spectra $d\mu_x$ and $d\mu_y$ have common point masses; if singular components are shared, the variance does not tend to zero, i.e., $\cov_{xy}$ cannot be consistently estimated. Furthermore, the auto-covariance function $\cov_x$ could only be consistently estimated by $\covest_x(\cdot;T)$ when the singular parts of the spectrum were modeled by fixed-amplitude components as in \eqref{eq:random_phase_model}. In order to explain this, we in Theorem~\ref{thm:auto_cov_general} present closed-form expressions for the auto-covariance estimate $\covest_x$ for the general model in \eqref{eq:general_model} that predict the particular cases in \eqref{eq:random_amplitude_model} and $\eqref{eq:random_phase_model}$. Theorem~\ref{thm:cross_cov} presents the corresponding expression for the finite-sample variance of $\covest_{xy}$. Furthermore, asymptotic expressions, i.e., valid as $T\to \infty$, for the respective quantities are presented in corresponding corollaries.

In the expressions for the estimator variance in the theorems, the continuous-time counterpart of the Fej\'er kernel \cite{StoicaM05},
\begin{equation}\label{eq:fejer}
\begin{aligned}
	\kernel_T(\freq)&\triangleq\int_{t=-T}^T (1-|t|/T) e^{i2\pi \freq t}dt\\
	&=\frac{2}{\freq^2T(2\pi)^2}(1-\cos(2\pi\freq T)),
\end{aligned}
\end{equation}
appears naturally.
The kernel $\kernel_T$ has several interesting properties that will be used in the derivations of the theorems; in particular, $\kernel_T$ acts as an approximate identity in convolutions. Letting $*$ denote convolution, we summarize these properties in the following proposition.
%
%
\begin{proposition}[Properties of $\kernel_T$]\label{prop:kernel}
 The following properties of $\kernel_T$ hold.
\begin{itemize}
	\item For all $T > 0$, $\kernel_T(\freq) \geq 0$ for all $\freq \in \RR$.
	\item For all $T > 0$, $\lim_{\freq\to0}\kernel_T(\freq) = T$.
	\item For all $\freq \neq 0$, $\lim_{T\to\infty} \kernel_T(\freq) = 0$.
	\item For all $T > 0$,
	$\int_{-\infty}^\infty \kernel_T(\freq)d\freq = 1.$
	\item For any $\Phi \in L_1(\RR)$, $\kernel_T*\Phi \to \Phi$ in $L_1$ as $T\to\infty$.
\end{itemize}
\end{proposition}
\begin{proof}
The first four properties are easily verified. For the last property, see, e.g., \cite[Chapter 2]{Hoffman62_banach}.
\end{proof}
With this, we are ready to state the theorems providing the theoretical estimator variances illustrated in Figures~\ref{fig:example_autocovariance} and \ref{fig:example_crosscovariance}.
%
%
\begin{theorem}\label{thm:auto_cov_general}
Let $x$ be as in \eqref{eq:general_model}. Then, the variance of the auto-covariance estimate is given by
\begin{align*}
&\var{\covest_x(\tau;T)}\\
&=\frac{1}{T}\int_\freq\int_\phi \kernel_T(\freq-\phi) d\mu_x(\freq)  d\mu_x(\phi) +\sum_{k=1}^{N_x}\left(\kurtrat_k-2\right)\xamp_k^4.
\end{align*}
\end{theorem}
\begin{proof}
See appendix.
\end{proof}
%
\begin{remark}
It may be noted that, perhaps surprisingly, the estimator variance does not depend on the lag $\tau$. This is due to the same averaging time, $T$, is used, irrespective of the lag. This is the relevant case for array processing; the covariance components constituting the array covariance matrix are all estimated from a common data length. All presented results may, in a straightforward manner, be modified as to account for other averaging times and time windows.
\end{remark}
It may be noted that the variance of the estimator is related to the concentration of mass as described by the spectrum $d\mu_x$ as well as to the kurtosis $\kappa_k$ of the amplitudes. The asymptotic variance is given in the following corollary.
%
\begin{corollary}\label{cor:auto_cov_general}
Let $x$ be as in \eqref{eq:general_model} and assume that $\Phi_x$ is continuous in the points $\freq_k^x$ for $k=1,\ldots, N_x$. Then, as $T\to \infty$,
\begin{align*}
	T\,\var{\covest_{x}(\tau;T)}  - \asympvarautoamp \to 0,
\end{align*}
where
\begin{align*}
	\asympvarautoamp = \int_{\freqband} \Phi_x(\freq)^2d\freq+ 2\sum_{k=1}^{N_x}\xamp_k^2 \Phi_x(\freq_k^x) + T\sum_{k=1}^{N_x}\left( \kurtrat_k - 1 \right)\xamp_k^4.
\end{align*}
\end{corollary}
\begin{proof}
The result follows directly by applying Theorem~\ref{thm:auto_cov_general} in the same way as Theorem~\ref{thm:cross_cov} is applied in the proof of Corollary~\ref{cor:asymptote}.
\end{proof}
%
Here, it may be noted that the variance of $\covest_x(\cdot;T)$ does not necessarily tend to zero as $T \to \infty$ when the process contains a sinusoidal component. In particular, the limiting variance is
\begin{align*}
	\lim_{T\to \infty}\frac{1}{T}\asympvarautoamp = \sum_{k=1}^{N_x}\left( \kurtrat_k - 1 \right)\xamp_k^4.
\end{align*}
Although no models \eqref{eq:general_model} are ergodic due to the stochastic amplitudes $\randamp_k^{(x)}$, one particular case allows for statistically consistent covariance estimates, as stated in the following corollary.
\begin{corollary}
As $T\to \infty$, the estimator variance $\var{\covest_x(\tau;T)}$ tends to zero if and only if all amplitudes follow the fixed magnitude model \eqref{eq:random_phase_model}.
\end{corollary}
\begin{proof}
It holds that (see, e.g., \cite{Mathis01}) $\kurtrat_k\geq 1$ with equality if and only if the random variables $\randamp_k^{(x)}$ are distributed according to \eqref{eq:random_phase_model}.
\end{proof}
Thus, as long as the magnitudes and not only the phases are random, the statistical moments cannot be estimated from single realizations of the process.  It may be noted that the asymptotic variance expression of Corollary~\ref{cor:auto_cov_general} for the special case $\kurtrat_k = 1$ coincides with that of \cite[Theorem 2]{Delmas01_47}. However, in the latter, the result holds for mixed-spectrum processes where the component with the density is a (not necessarily Gaussian) moving average process, whereas the result derived herein concerns Gaussian processes with general densities.

The results concerning the estimation of the auto-covariance of a signal can be extended to the case where the cross-covariance is estimated. The following theorem holds.
%
\begin{theorem}\label{thm:cross_cov}
Let $x$ and $y$ be as in \eqref{eq:general_model}. Then, the variance of the cross-covariance estimate is given by 
\begin{align*}
\var{\covest_{xy}(\tau;T)} &= \expop\left(\abs{\covest_{xy}(\tau;T)}^2\right) -\abs{\expop\left(\covest_{xy}(\tau;T)\right)}^2\\
&= \frac{1}{T}\int_\freq\int_\phi \kernel_T(\freq-\phi) d\mu_x(\freq)  d\mu_y(\phi).
\end{align*}
\end{theorem}
\begin{proof}
See appendix.
\end{proof}
As can be seen from Theorem~\ref{thm:cross_cov}, $\covest_{xy}(\tau;T)$ is an unbiased estimator of $\cov_{xy}(\tau)$, with a variance that depends on the overlap of the spectra $d\mu_x$ and $d\mu_y$. However, it is not necessarily a consistent estimator, as shown in Corollary~\ref{cor:asymptote}.
%
%
%
\begin{corollary} \label{cor:asymptote}
Let $x$ and $y$ be as in \eqref{eq:general_model} and assume that $\Phi_y$ is continuous in the points $\freq_k^x$ for $k=1,\ldots, N_x$, and that $\Phi_x$ is continuous in the points $\freq_\ell^y$ for $\ell=1,\ldots, N_y$. Then, as $T\to \infty$,
\begin{align*}
	T\,{\rm Var}\left(\covest_{xy}(\tau;T)\right)  - \asympvarcross \to 0,
\end{align*}
where
\begin{align*}
\asympvarcross &=  \int_{\freqband} \Phi_x(\freq)  \Phi_y(\freq)d\freq+ \sum_{k=1}^{N_x}\xamp_k^2 \Phi_y(\freq_k^x)+ \sum_{k=1}^{N_y}\yamp_k^2 \Phi_x(\freq_k^y)\\
&\quad+T\sum_{k=1}^{N_x}\sum_{\ell=1}^{N_y}  \xamp_k^2\yamp_\ell^2 \chi_{\{\freq_k^x=\freq_\ell^y\}}
\end{align*}
and where $\chi$ is the characteristic function.
\end{corollary}
\begin{proof}
See the appendix.
\end{proof}
%
%
It may be noted that if the two signals share sinusoidal components, then 
$\asympvarcross$ is not bounded as $T \to \infty$. Also in general, the asymptotic variance, $\frac{1}{T} \asympvarcross$ tends to
\begin{align*}
	\lim_{T\to\infty} \frac{1}{T} \asympvarcross = \sum_{k=1}^{N_x}\sum_{\ell=1}^{N_y}  \xamp_k^2\yamp_\ell^2 \chi_{\{\freq_k^x=\freq_\ell^y\}},
\end{align*}
which is strictly positive if any sinusoidal frequencies are common. Interestingly, the issue of ergodicity is only apparent if point masses are shared; if all sinusoidal frequencies are distinct, the estimator variance tends to zero.
%
%
%

The results of Theorems~\ref{thm:auto_cov_general} and \ref{thm:cross_cov}, have implications for inference for array processing applications. In fact, if mixed-spectrum processes are considered and the signals are generated according to \eqref{eq:random_phase_model}, then the array covariance function can be consistently estimated as long as the individual processes do not share any sinusoidal components. In contrast, if the model \eqref{eq:random_amplitude_model}, or indeed any model other than \eqref{eq:random_phase_model}, is used, the array covariance cannot be estimated as the estimates of the auto-covariance functions do not converge to their expectations. 
As we will see next, the results of the presented theorems have implications for how to approximate processes with absolutely continuous spectra. Although the presented results hold for the general model \eqref{eq:general_model}, we give particular attention to the models in \eqref{eq:random_amplitude_model} and \eqref{eq:random_phase_model} as these two models possess special properties that will become clear in the following exposition.
\section{Approximations of continuous spectra}
Consider the problem of generating realizations from a Gaussian process $x$ with a continuous spectral density $\Phi\in C(\freqband)$ as to, e.g., simulate broadband array signals. To this end, one may\footnote{Alternatively, one could, e.g, consider sampling from an ARMA process with the correct spectral shape in the band $\freqband$, together with appropriate bandpass filtering.} approximate the target signal using processes with completely singular spectra, i.e., sinusoidal models. As formalized in the following theorem, any member of the model family \eqref{eq:general_model} allows for approximating a process with a spectral density $\Phi$.
%
%
\begin{theorem} \label{thm:conv_in_dist}
Let $x$ be a band-limited Gaussian WSS process with a continuous spectrum $\Phi\in C(\freqband)$ 
with support in $\freqband = [\freq_c - B/2,\freq_c+B/2]$, where $\freq_c$ is the center frequency and $B$ is the bandwidth. Define the sequence of approximating processes
\begin{align} \label{eq:approxseq}
	x^{(n)}(t) = \sum_{k=1}^n \sqrt{ \frac{B}{n}\Phi\left(\freq_k^{(n)}\right)} e^{i2\pi\freq^{(n)}_k t}\randamp^{(n)}_k,
\end{align}
where the frequency points $\freq_k^{(n)} = \freq_c + B\frac{k-1-n/2}{n}$ for $k=1,\ldots,n$, define a uniform grid on  $\mathcal{I}_B$.
Here, $\randamp_k^{(n)}$ are independent identically distributed circularly symmetric random variables such that $\expect{\abs{\randamp_k^{(n)}}^2} = 1$, and with finite kurtosis parameter $\kappa$.
Then, the sequence of processes $\left\{ x^{(n)} \right\}_n$ converges  to $x$ in distribution when $n\to \infty$.
\end{theorem}
\begin{proof}
	See appendix.
\end{proof}
%

\begin{remark}
The result of Theorem~\ref{thm:conv_in_dist} may be generalized in a straightforward manner to spectral densities $\Phi$ that are piecewise continuous with a finite number of discontinuities. 
\end{remark}
Thus, constructing approximations from \eqref{eq:general_model} yields processes that converge in distribution to any given process $x$ when the number of point masses, $n$, tends to infinity. To see that this is consistent with Theorem~\ref{thm:auto_cov_general}, i.e., that the covariance estimator for such approximation should in the limit behave as for the process $x$, it may be noted that, as a consequence of Lemma~\ref{lemma:covariance_convergence} (see the appendix),
\begin{align*}
	d\mu^{(n)}(\freq) = \sum_{k=1}^n  \frac{B}{n}\Phi\left(\freq_k^{(n)}\right) \delta\left(\freq-\freq_k^{(n)}\right)\weakstarconv \Phi(\freq),
\end{align*}
where $\weakstarconv$ denotes weak$^*$ convergence. Then, as $f_T$ is continuous for any finite $T$,
\begin{align*}
	&\int_\freq\int_\phi f_T(\freq-\phi) d\mu^{(n)}(\freq)  d\mu^{(n)}(\phi) \\
	&\qquad\to \int_\freq\int_\phi f_T(\freq-\phi) \Phi(\freq) \Phi(\phi)d\freq d\phi,
\end{align*}
as $n \to \infty$. Furthermore, the term related to the sinusoidal amplitudes in Theorem~\ref{thm:auto_cov_general} is given by
\begin{align*}
	A(n) &\triangleq\sum_{k=1}^n \left(\frac{B}{n}\Phi\left(\freq_k^{(n)}\right)\right)^2 = \frac{1}{n^2}\sum_{k=1}^n \left(B\Phi\left(\freq_k^{(n)}\right)\right)^2,\\
&\le \frac{B^2}{n}\max_{\freq\in \freqband}\Phi\left(\freq\right)^2,
\end{align*}
and since $\Phi$ is bounded,\footnote{Note that $\Phi$ is continuous on a compact interval.} we have that $A(n) \to 0$ as $n \to \infty$. Thus, Theorem~\ref{thm:auto_cov_general} predicts the correct limiting behavior of approximations constructed according to Theorem~\ref{thm:conv_in_dist}. It may however be noted that this is only valid if $T$ is fixed. In order to obtain a description of the behavior of the covariance estimators in the asymptotic regime, i.e., when both $n$ and $T$ are large, we consider the following theorem.
\begin{theorem}\label{thm:approximation_variance}
Let $x^{(n)}$ be processes with spectra $d\mu^{(n)}$, approximating the process $x$ with a continuous spectrum $\Phi$, as in Theorem~\ref{thm:conv_in_dist}, and let $\covest_{x}^{(n)}(\tau;T)$ be the corresponding covariance estimate as in \eqref{eq:autocov_est_a}. Let $n\to \infty, T\to \infty$ with $T\frac{B}{n}\to \asympratio$. Then, 
\begin{align*}
T\,\var{\covest_{x}^{(n)}(\tau;T)} \to \left( (\kurtrat-1)\asympratio+ \rho(\asympratio)\right) \int_{\freqband}\Phi(\freq)^2d\freq,
\end{align*}
where
\begin{align*}
\rho(\asympratio)=\frac{ \breve{\asympratio}(1- \breve{\asympratio})}{\asympratio}
\end{align*}
and $\breve{\asympratio}$ is the decimal part\footnote{
That is, $\breve{\asympratio} = \asympratio -  \lfloor\asympratio\rfloor$ where $\lfloor\asympratio\rfloor$ denotes the integer part of $\asympratio$.} 
of $\asympratio$ and where $\kurtrat$ is the common kurtosis parameter. 
\end{theorem}
%
\begin{figure}[t]
        \centering
            \includegraphics[width=.46\textwidth]{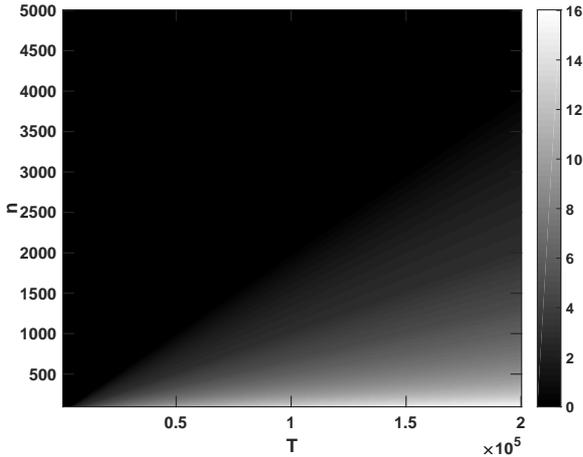}
           \caption{The asymptotic variance factor $\asympratio + \rho(\asympratio)$, where $\asympratio = T\frac{B}{n}$, for the Gaussian amplitude approximation in Theorem~\ref{thm:conv_in_dist}, presented as $10\log_{10} \left(\asympratio + \rho(\asympratio)\right)$, for large $T$ and $n$, as given by Theorem~\ref{thm:approximation_variance}. Here, the bandwidth is $B = 2\times10^{-2}$.} 
            \label{fig:Gaussian_amp_factor}
\vspace{-2mm}\end{figure}
%
%
\begin{remark}
The proof of Theorem~\ref{thm:approximation_variance} is based on the fact that 
\begin{align*}
	\sum_{m=-\infty}^\infty \kernel_\asympratio(m) = \asympratio +\rho(\asympratio),
\end{align*}
where $\kernel_\asympratio$ is the Fej\'er kernel defined in \eqref{eq:fejer}. Therefore, the variance expression connects directly to the sampling of the frequency band used in the approximation.
\end{remark}
\begin{remark}
It may be noted that for $\asympratio \in (0,1]$, $\rho(\asympratio) = 1-\asympratio$, and thus $\asympratio + \rho(\asympratio) \equiv 1$. Thus, for $n\geq BT$, approximations constructed from components with $\kurtrat = 2$, that is, corresponding to Gaussian variables, behave as the target process, i.e., the asymptotic estimator variance is $\frac{1}{T}\int_{\freqband}\Phi(\theta)^2d\theta$.
Furthermore, it is readily verified that $\asympratio + \rho(\asympratio)$ is continuous, monotone increasing, and that $\lim_{\asympratio \to \infty} \rho(\asympratio) = 0$, implying that the variance of such approximations is strictly greater than that of the target process for all $\asympratio>1$. It also directly follows that for approximations with $\kurtrat > 2$, the variance of the approximation is always strictly greater than that of the target process.
Conversely, for $\kurtrat \in [1,2)$, the variance of the approximation is strictly smaller than that of the target for $\asympratio \in (0,1]$. Thus, for this class of approximations to mimic the target process, it is required that $n\gg BT$, as \mbox{$\lim_{\asympratio \to 0} \rho(\asympratio) = 1$}.
For the special case of the fixed magnitude model with $\kurtrat = 1$, the variance is exactly zero for integer $\gamma$ as $\rho(\asympratio) = 0$ for $\asympratio \in \RN$. This is also the only model for which the variance tends to zero as $\asympratio \to \infty$.
\end{remark}
The parameter $\asympratio$ is the product of the measurement duration $T$ and the frequency resolution $B/n$ of the discretization. According to Theorem~\ref{thm:approximation_variance}, the finer the frequency resolution is in relation to the measurement time, the more the singular approximations behave as a process with a spectral density, whereas the properties pertaining from the singular spectra become more apparent when $T$ is large in relation to $B/n$.
Thus, even though all models as $\asympratio\to 0$ have the same statistical behavior as the target process $x$, they can for $\asympratio > 0$, i.e., the number of components only moderately large compared to the measurement time $T$, behave dramatically different. Specifically, if using approximations with $\kurtrat < 2$, and, in particular, the fixed magnitude model, with finite $n$ as to, e.g., generate test data for array processing, one runs the risk of obtaining too well-behaved results as compared to if an actual process $x$ with a spectral density would be used. As an illustration of this, the impact of model choice for DoA estimation with broadband sources will be provided in the numerical section.
%
\begin{figure}[t]
        \centering
            \includegraphics[width=.46\textwidth]{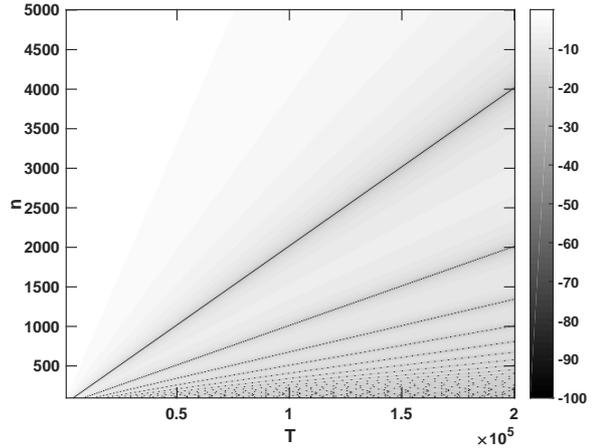}
           \caption{The asymptotic variance factor $\rho(\asympratio)$, where $\asympratio = T\frac{B}{n}$, for the fixed magnitude approximation in Theorem~\ref{thm:conv_in_dist}, presented as $10\log_{10} \left(\rho(\asympratio)\right)$, for large $T$ and $n$, as given by Theorem~\ref{thm:approximation_variance}. Here, the bandwidth is $B = 2\times10^{-2}$.} 
            \label{fig:fixed_amp_factor}
\vspace{-2mm}\end{figure}
%
%
%
\begin{remark}
It may be noted that the result of Theorem~\ref{thm:approximation_variance} may be generalized in a straightforward manner to components with different kurtosis parameter $\kurtrat$. Specifically, as this parameter then is a function of the frequency of that component, the corresponding variance expression becomes
\begin{align*}
	\int_{\freqband}\left( (\kurtrat(\freq)-1)\asympratio+ \rho(\asympratio)\right)\Phi(\freq)^2d\freq.
\end{align*}
Indeed, also Theorem~\ref{thm:conv_in_dist} holds in this case, as only independence of the components, and not equality of the distributions, is needed.
\end{remark}
\section{Numerical illustrations}
In this section, we illustrate the results of the presented theorems by numerical examples. To illustrate the significance of the derived results, we will throughout the examples consider two particular contrasting models, namely the Gaussian amplitude model in \eqref{eq:random_amplitude_model} and the fixed magnitude model in \eqref{eq:random_phase_model}.
\subsection{Asymptotics for singular approximations}
To demonstrate the results of Theorem~\ref{thm:approximation_variance}, Figures~\ref{fig:Gaussian_amp_factor} and \ref{fig:fixed_amp_factor} present the scaling factors $\asympratio + \rho(\asympratio)$ and $\rho(\asympratio)$, corresponding to the Gaussian amplitude and fixed magnitude approximations, respectively, for varying values of $n$ and $T$. Here, the bandwidth is fixed to $B = 2\times 10^{-2}$. As can be seen from Figure~\ref{fig:Gaussian_amp_factor}, the scaling factor for the Gaussian amplitude model is bounded from below by $1$, which is attained for $\asympratio \leq 1$, corresponding to values of $n$ that are large relative to $T$. In contrast, as can be seen in Figure~\ref{fig:fixed_amp_factor}, the scaling factor only asymptotically approaches $1$ from below as $\asympratio \to 0$, i.e., as $n$ grows in relation to $T$. Furthermore, it may be noted that the scaling factor is exactly zero for finite $n$ and $T$ corresponding to integer values of $\asympratio$.
\subsection{Estimator variance for singular approximations}
To illustrate the behavior of the singular approximations of processes with spectral densities, consider the spectrum
\begin{align}\label{eq:spectrum_white_in_band}
	\Phi(\freq) = \begin{cases}
		\frac{1}{B} & \freq \in \freqband\\
		0 & \freq \notin \freqband,
	\end{cases}
\end{align}
where the center frequency and bandwidth of $\freqband$ are $\omega_c = 1$ and $B = 10^{-2}$, respectively. We then approximate $\Phi$ according to Theorem~\ref{thm:conv_in_dist} using both the fixed magnitude model and the model with Gaussian amplitudes. In both cases we, consider approximations with $n = 100$ and $n = 1000$ components.
%
%
\begin{figure}[t]
        \centering
            \includegraphics[width=.465\textwidth]{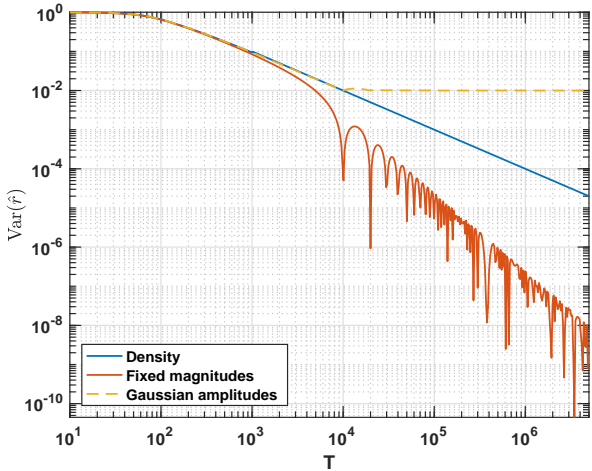}
           \caption{Variance for the covariance estimator for a covariance function corresponding to an absolutely continuous process, as well as for singular approximations according to Theorem~\ref{thm:conv_in_dist} for $n = 100$.} 
            \label{fig:discrete_approximation_n100}
\vspace{-2mm}\end{figure}
%
%
Letting $r_a$ and $r_s^{(n)}$ denote the covariance functions corresponding to $\Phi$ and a singular approximation with $n$ components, respectively, where it may be noted that the covariance function for the two singular approximations are identical, let $\errorfunc: \RR\times\RN \to \RR_+$ be defined as
\begin{align*}
	\errorfunc(\tau,n) \triangleq \sqrt{\frac{\int_{0}^\tau \abs{\cov_a(t)-\cov_s^{(n)}(t)}^2dt}{\int_0^\tau \abs{\cov_a(t)}^2dt}},
\end{align*}
i.e., the relative $L_2$ error when considering the covariance up to lag $\tau$. In this case, considering a maximum lag of $\tau = 100$, we have $\errorfunc(100,100) = 1\times 10^{-6}$ and $\errorfunc(100,1000) = 2\times 10^{-7}$ for the approximations with $n = 10$ and $n = 100$ components, respectively. With this, Figures~\ref{fig:discrete_approximation_n100} and \ref{fig:discrete_approximation_n1000} display the variance of the estimators $\covest_a$ and $\covest_s^{(n)}$ as a function of the measurement duration $T$, for $n = 100$ and $n = 1000$, respectively. The estimator variances are computed according to Theorem~\ref{thm:auto_cov_general}. As can be seen, the variance corresponding the fixed magnitude model is consistently lower than that of the process with a density. Furthermore, even though the Gaussian amplitude model mimics the target process perfectly for $T\leq \frac{n}{B}$, the variance does not tend to zero as $T \to\infty$. In fact, the variance stabilizes for $T \geq 10^{4}$ and $T \geq 10^{5}$ for $n = 100$ and $n = 1000$, respectively, corresponding to $\asympratio \leq 1$, as predicted by Theorem~\ref{thm:approximation_variance}.
%
%
\begin{figure}[t]
        \centering
            \includegraphics[width=.46\textwidth]{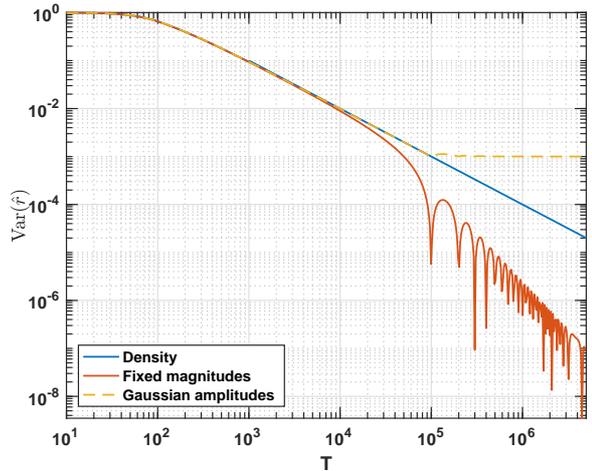}
           \caption{Variance for the covariance estimator for a covariance function corresponding to an absolutely continuous process, as well as for singular approximations according to Theorem~\ref{thm:conv_in_dist} for $n = 1000$.} 
            \label{fig:discrete_approximation_n1000}
\vspace{-2mm}\end{figure}
%
%
%
%
\subsection{Implications for array processing: DoA estimation}
As noted, the singular approximations \eqref{eq:approxseq} in Theorem~\ref{thm:conv_in_dist} differ considerably in terms of their behavior in covariance estimation depending on the kurtosis parameter $\kurtrat$ of the distribution of the stochastic amplitudes $\randamp_k^{(n)}$.
To illustrate the implication of this for array processing, we consider a simple DoA estimation example where two sources, both with spectra as in \eqref{eq:spectrum_white_in_band}, i.e., spectral densities, with bandwidth $B = 10^{-3}$ and center frequency $\freq_c = 0.25$, impinge from angles $-5$ degrees and $10$ degrees, respectively, on a uniform linear array consisting of $10$ sensors with inter-sensor spacing just below half of the highest frequency in the support of \eqref{eq:spectrum_white_in_band}.
As before, we as contrasting examples consider the fixed magnitude model and Gaussian amplitude model, corresponding to $\kurtrat = 1$ and $\kurtrat = 2$, respectively.
We add a spatially and temporally white Gaussian sensor noise to the sensor signals, yielding a signal-to-noise-ratio (SNR) of 10 dB. The array covariance matrix is estimated as the sample covariance matrix, averaging $T = 5\times10^{5}$ consecutive array snapshots. It may here be noted that the snapshots are not independent as consecutive samples are considered.
The spatial spectrum is estimated by integrating the narrowband Capon\footnote{As the Capon spectral estimator is non-linear in the array covariance matrix estimates, the results from Theorem~\ref{thm:approximation_variance} can only be expected to hold qualitatively.} spatial spectrum \cite{Capon69} over the frequency band $\freqband$.
%
%
%
\begin{figure}[t]
        \centering
            \includegraphics[width=.47\textwidth]{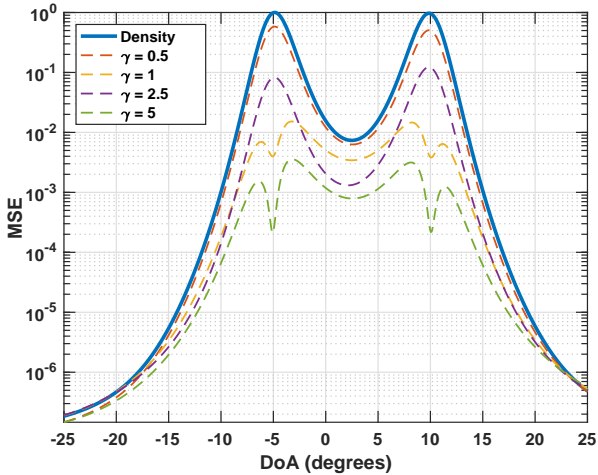}
           \caption{The MSE of the Capon spatial spectrum, integrated over $\freqband$, for two signals with spectral densities, both of bandwidth $B = 10^{-3}$, impinging from $-5$ degrees and $10$ degrees, when approximated according to \eqref{eq:approxseq} using the fixed magnitude model for a measurement duration of $T = 5\times 10^{5 }$. The number $n$ of discretization components vary with the parameter $\asympratio$ from Theorem~\ref{thm:approximation_variance}.} 
            \label{fig:capon_fixed_amp_two_sources}
\vspace{-2mm}\end{figure}
%
%
This is performed for singular approximations of the signals with spectral densities, constructed according to \eqref{eq:approxseq} in Theorem~\ref{thm:conv_in_dist}. For these approximations, we consider varying the parameter $\asympratio$, and thereby $n$ as $T$ and $B$ are fixed, between $\asympratio = 0.5$, corresponding to $n = 1000$, and $\asympratio = 5$, corresponding to $n = 100$. As to avoid the problem of non-vanishing variance observed in Corollary~\ref{cor:asymptote}, the singular components of the second source are shifted in frequency by $B/2n$ as to avoid any overlap. The procedure is repeated in 100 Monte Carlo simulations. The per-angle mean squared error (MSE) for the estimated spatial spectra\footnote{The reference is the corresponding Capon spectrum computed using the exact array covariance matrix. It may be noted that the approximations incur a bias due to the discretization. However, for the considered values of $\asympratio$, the squared bias is two orders of magnitude smaller than the variance corresponding to the process with a density.} are presented in Figures~\ref{fig:capon_fixed_amp_two_sources} and \ref{fig:capon_Gauss_amp_two_sources} for the fixed magnitude and Gaussian amplitude approximations, respectively. As reference, the corresponding MSE of the estimated spatial spectrum for the target process with spectral density, generated by bandpass filtering white noise using a Butterworth filter with passband $\freqband$, is also presented. It may be noted that all values are normalized by the largest per-angle MSE corresponding to the filtered process. As can be seen in Figure~\ref{fig:capon_fixed_amp_two_sources}, the MSE of the spatial spectrum corresponding to the fixed magnitude approximation is lower than that of the filtered process for all considered values of $\asympratio$. One may here recall from Theorem~\ref{thm:approximation_variance} that it is required that $\asympratio \to 0$ for the variance of the covariance estimate to converge to that of the process with spectral density. Furthermore, a drop in the MSE may be observed for integer values of $\asympratio$. For these values, the autocovariances for the two sources are perfectly estimated (c.f. Theorem~\ref{thm:approximation_variance}), and the variability stems from the sensor noise and the non-zero variance of the estimates of the sources' cross-covariance. It may here be noted that the MSE does not strictly decrease with increasing $\asympratio$ as $\rho$ is not monotone. In contrast, the MSE for the Gaussian amplitude approximation coincides with that of the filtered process for $\asympratio \leq 1$, whereas being higher for $\asympratio>1$, in accordance with Theorem~\ref{thm:approximation_variance}.
%
%
\begin{figure}[t]
        \centering
            \includegraphics[width=.47\textwidth]{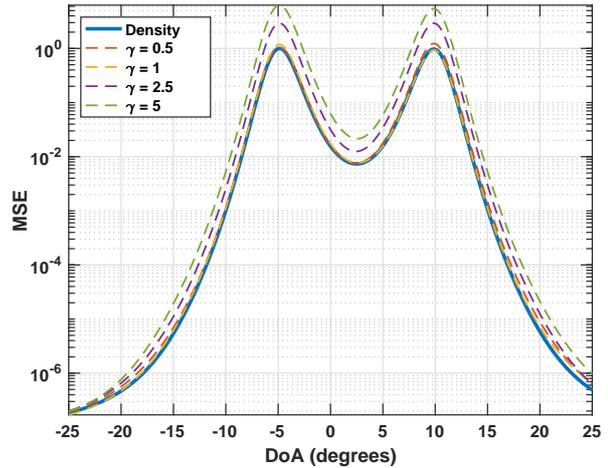}
           \caption{The MSE of the Capon spatial spectrum, integrated over $\freqband$, for two signals with spectral densities, both of bandwidth $B = 10^{-3}$, impinging from $-5$ degrees and $10$ degrees, when approximated according to \eqref{eq:approxseq} using the Gaussian amplitude model for a measurement duration of $T = 5\times 10^{5 }$. The number $n$ of discretization components vary with the parameter $\asympratio$ from Theorem~\ref{thm:approximation_variance}.} 
            \label{fig:capon_Gauss_amp_two_sources}
\vspace{-2mm}\end{figure}
%
%
\section{Conclusions}
In this work, we have derived exact finite-sample as well as asymptotic large-sample expressions for the statistical variance of covariance function estimates for mixed-spectrum signals. As has been shown, the statistical properties of such estimates differ considerably depending on how the singular part of the spectrum is modeled. Furthermore, for singular approximations of processes with continuous spectra, we have presented asymptotic regime results for the covariance estimator variance when both the measurement time and the number of approximating components tend to infinity. As has been illustrated, the difference in variability of the covariance estimates corresponding to the different approximations have a considerable impact on the statistical performance of array processing algorithms.
%
%
\appendix
\section{Proofs}
%
\begin{proof}[Proof of Theorem \ref{thm:auto_cov_general}]
As to simplify notation, let $x(t) = x_a(t) + x_s(t)$, where $x_s$ denotes the sinusoidal part of $x$. Furthermore, let $\cov_a$ and $\cov_s$ be the covariance functions of $x_a$ and $x_s$, respectively. Then, $\expop\left( \covest_x(\tau;T) \right) = \cov_x(\tau) = \cov_a(\tau) + \cov_s(\tau)$. Furthermore, 
\begin{align*}
	&\expect{\abs{\covest_x(\tau;T)}^2} \\
	&=\expop\left[\left(\frac{1}{T}\int_{t=0}^T x(t)\overline{x(t-\tau)}dt\right)\left(\frac{1}{T}\int_{t=0}^T \overline{x(t)}x(t-\tau)dt\right)\right] \\
	&= \frac{1}{T^2}\int_{t=0}^T\int_{\sigma=0}^T \expect{ x(t)\overline{x(\sigma)}x(\sigma-\tau)\overline{x(t-\tau)}}dt d\sigma.
\end{align*}
As $x_a$ and $x_s$ are independent, expanding the product yields
\begin{align*}
&\expect{ x(t)\overline{x(t-\tau)}\; \overline{x(\sigma)}x(\sigma-\tau)}\\
&=\expect{ x_a(t)\overline{x_a(t-\tau)}\; \overline{x_a(\sigma)}x_a(\sigma-\tau)}\\
&\quad+\expect{ x_s(t)\overline{x_s(t-\tau)}\; \overline{x_s(\sigma)}x_s(\sigma-\tau)}\\
&\quad+\expect{ x_a(t)\overline{x_a(t-\tau)}}\; \expect{\overline{x_s(\sigma)}x_s(\sigma-\tau)}\\
&\quad+\expect{x_s(t)\overline{x_s(t-\tau)}}\; \expect{\overline{x_a(\sigma)}x_a(\sigma-\tau)}\\
&\quad+\expect{x_a(t)\overline{x_a(\sigma)}}\; \expect{\overline{x_s(t-\tau)}x_s(\sigma-\tau)}\\
&\quad+\expect{ x_s(t)\overline{x_s(\sigma)}}\; \expect{\overline{x_a(t-\tau)}x_a(\sigma-\tau)}\\
&=\expect{ x_a(t)\overline{x_a(t-\tau)}\; \overline{x_a(\sigma)}x_a(\sigma-\tau)}\\
&\quad+\expect{ x_s(t)\overline{x_s(t-\tau)}\; \overline{x_s(\sigma)}x_s(\sigma-\tau)}\\
&\quad+\cov_a(\tau)\cov_s(-\tau)+\cov_s(\tau)\cov_a(-\tau) \\
&\quad+ \cov_a(t-\sigma)\cov_s(\sigma-t) + \cov_s(t-\sigma)\cov_a(\sigma-t) \\
&=\expect{ x_a(t)\overline{x_a(t-\tau)}\; \overline{x_a(\sigma)}x_a(\sigma-\tau)}\\
&\quad+\expect{ x_s(t)\overline{x_s(t-\tau)}\; \overline{x_s(\sigma)}x_s(\sigma-\tau)}\\
&\quad+\abs{\cov_a(t-\sigma)+\cov_s(t-\sigma)}^2 + \abs{\cov_a(\tau)+\cov_s(\tau)}^2 \\
&\quad - \abs{r_a(t-\sigma)}^2 - \abs{r_s(t-\sigma)}^2- \abs{r_a(\tau)}^2 - \abs{r_s(\tau)}^2.
\end{align*}
Furthermore, as $x_a$ is Gaussian, circularly symmetric, and zero-mean,
\begin{align*}
&\expect{ x_a(t)\overline{x_a(t-\tau)}\; \overline{x_a(\sigma)}x_a(\sigma-\tau)}\\
&=\expect{ x_a(t)\overline{x_a(t-\tau)}}\expect{ \overline{x_a(\sigma)}x_a(\sigma-\tau)}\\
&\quad+\expect{ x_a(t)\overline{x_a(\sigma)}}\; \expect{\overline{x_a(t-\tau)}x_a(\sigma-\tau)}\\
&=\cov_a(\tau)\cov_a(-\tau)+\cov_a(t-\sigma)\cov_a(\sigma-t) \\
&= \abs{\cov_a(\tau)}^2+\abs{\cov_a(t-\sigma)}^2.
\end{align*}
Thus,
\begin{equation} \label{eq:fourth_moment}
\begin{aligned}
&\expect{ x(t)\overline{x(t-\tau)}\; \overline{x(\sigma)}x(\sigma-\tau)}\\
&= \expect{ x_s(t)\overline{x_s(t-\tau)}\; \overline{x_s(\sigma)}x_s(\sigma-\tau)}\\
&\quad+\abs{\cov_x(t-\sigma)}^2 + \abs{\cov_x(\tau)}^2 - \abs{r_s(t-\sigma)}^2- \abs{r_s(\tau)}^2.
\end{aligned}
\end{equation}
To compute the fourth moment of $x_s$, consider four time points $t_1$, $t_2$, $t_3$, and $t_4$. Then,
\begin{align} \label{eq:fourth_term_stoch}
x_s(t_1)x_s(t_2)\overline{x_s(t_3)}\overline{x_s(t_4)}\!=\!\sum_{k,\ell,m,n}\! \randamp_k \randamp_\ell \overline{\randamp_m \randamp_n} e^{i\xi_{k,\ell,m,n}}
\end{align}
where 
\begin{align*}
	\xi_{k,\ell,m,n} \!=\! 2\pi(\freq_k t_1\!+\! \freq_\ell t_2\!-\! \freq_m t_3\!-\!\freq_n t_4),
\end{align*}
and where the superscript of $\randamp_k = \randamp_k^{(x)}$ has been suppressed for notational brevity. Since all amplitudes $\randamp_k$ are independent and circular symmetric, the expectation of the terms in \eqref{eq:fourth_term_stoch} are only non-zero when $k=m$ and $\ell=n$, or $k=n$ and $\ell=m$. Thus,
\begin{align*}
&\expect{ x_s(t_1)x_s(t_2)\overline{x_s(t_3)}\overline{x_s(t_4)}}\\
&\qquad\!=\!\sum_k\sum_\ell\!\xamp_k^2\xamp_\ell^2\!\Big(e^{i2\pi\freq_k (t_1-t_4)+i2\pi\freq_\ell(t_2-t_3)}\!\\
&\qquad+\!e^{i2\pi\freq_k (t_1-t_3)+i2\pi\freq_\ell(t_2-t_4)}\Big)\\
&\qquad+\sum_k (\kurtrat_k-2) \xamp_k^4 e^{i2\pi\freq_k (t_1+t_2-t_3-t_4)},
\end{align*}
where we recall that $\kurtrat_k = \expect{\abs{z_k}^4}/\expect{\abs{z_k}^2}^2$. Plugging in the corresponding time lags $t_1=t$, $t_2=\sigma-\tau$, $t_3=t-\tau$, and $t_4=\sigma$, the double sum becomes
\begin{align*}
&\sum_k\sum_\ell\!\xamp_k^2\xamp_\ell^2\!\left(e^{i2\pi\freq_k (t-\sigma)+i2\pi\freq_\ell(\sigma-t)}\!+\!e^{i2\pi\freq_k (\tau)+i2\pi\freq_\ell(-\tau)}\right)\\
&= \sum_k \xamp_k^2e^{i2\pi\freq_k (t-\sigma)}\sum_\ell\!\xamp_\ell^2e^{i2\pi\freq_\ell(\sigma-t)} \\
&\quad+\sum_k \xamp_k^2e^{i2\pi\freq_k \tau}\sum_\ell\!\xamp_\ell^2e^{-i2\pi\freq_\ell \tau} \\
&= \abs{\cov_s(t-\sigma)}^2 + \abs{\cov_s(\tau)}^2.
\end{align*}
Noting that $t_1+t_2+t_3+t_4 = 0$, we get
\begin{align*}
&\expect{ x_s(t)\overline{x_s(t-\tau)}\overline{x_s(\sigma)}x_s(\sigma-\tau)} \\
&\quad= \abs{\cov_s(t-\sigma)}^2 + \abs{\cov_s(\tau)}^2+\sum_k (\kurtrat_k-2)\xamp_k^4.
\end{align*}
Inserting this expression in \eqref{eq:fourth_moment} yields
\begin{equation*}
\begin{aligned}
&\expect{ x(t)\overline{x(t-\tau)}\; \overline{x(\sigma)}x(\sigma-\tau)} \\
&\qquad= \abs{\cov_x(t-\sigma)}^2 + \abs{\cov_x(\tau)}^2 + \sum_k (\kurtrat_k-2)\xamp_k^4.
\end{aligned}
\end{equation*}
Thus, since $\expect{\covest_{x}(\tau;T) } = \cov_{x}(\tau)$, the variance is
\begin{align*}
&\var{\covest_{x}(\tau;T)} = \expect{\abs{\covest_{x}(\tau;T)}^2} -\abs{\cov_{x}(\tau)}^2\\
 &= \frac{1}{T^2}\int_{t=0}^T\int_{\sigma=0}^T\abs{\cov_x(t-\sigma)}^2dt d\sigma + \sum_k (\kurtrat_k-2)\xamp_k^4.
 \end{align*}
Finally, the integral is given by
\begin{align*}
&\frac{1}{T^2}\int_{t=0}^T\int_{\sigma=0}^T\cov_x(t-\sigma) \cov_x(\sigma-t)dt d\sigma \\
&=\frac{1}{T^2}\!\int_{t=0}^T\!\int_{\sigma=0}^T\!\int_\theta \!e^{i2\pi\theta(t-\sigma)}d\mu_x(\theta) \!\!\int_\phi e^{i2\pi\phi(\sigma-t)}d\mu_x(\phi)dt d\sigma\\
&=\frac{1}{T^2}\int_\theta\int_\phi\left(\int_{t=0}^T\int_{\sigma=0}^T e^{i2\pi(\theta-\phi)(t-\sigma)}dt d\sigma\right)d\mu_x(\theta)  d\mu_x(\phi)\\
&=\frac{1}{T}\int_\theta\int_\phi\left(\int_{t=-T}^T (1-|t|/T) e^{i2\pi(\theta-\phi)t}dt\right)d\mu_x(\theta)  d\mu_x(\phi)\\
&= \frac{1}{T}\int_\theta\int_\phi \kernel_T(\theta-\phi)d\mu_x(\theta)  d\mu_x(\phi).
\end{align*}
\end{proof}
%
%
%
%
\begin{proof}[Proof of Theorem \ref{thm:cross_cov}]

First note that
\begin{align*}
 &\expect{\abs{\covest_{xy}(\tau;T)}^2} \\
&= \expop\left[\left(\frac{1}{T}\int_{t=0}^T x(t)\overline{y(t-\tau)}dt\right)\left(\frac{1}{T}\int_{\sigma=0}^T \overline{x(\sigma)}y(\sigma-\tau)d\sigma\right)\right]\\
&=\frac{1}{T^2}\int_{t=0}^T\int_{\sigma=0}^T\expect{x(t)\overline{x(\sigma)}}\expect{y(\sigma-\tau)\overline{y(t-\tau)}}dt d\sigma\\
&+\frac{1}{T^2}\int_{t=0}^T\int_{\sigma=0}^T \expect{x(t)\overline{y(t-\tau)}}\expect{y(\sigma-\tau)\overline{x(\sigma)}}dt d\sigma\\
&= \frac{1}{T^2}\int_{t=0}^T\int_{\sigma=0}^T\left(\cov_x(t-\sigma) \cov_y(\sigma-t)+r_{xy}(\tau)r_{yx}(\tau)\right)dt d\sigma\\
&= \frac{1}{T^2}\int_{t=0}^T\int_{\sigma=0}^T\cov_x(t-\sigma) \cov_y(\sigma-t)dt d\sigma+ |r_{xy}(\tau)|^2.
\end{align*}
Next, since
$\expect{\covest_{xy}(\tau;T) } = \cov_{xy}(\tau)$, the variance is
\begin{align*}
&\var{\covest_{xy}(\tau;T)} = \expect{\abs{\covest_{xy}(\tau;T)}^2} -\abs{\expect{\covest_{xy}(\tau;T)}}^2\\
&= \expect{\abs{\covest_{xy}(\tau;T)}^2} -\abs{\cov_{xy}(\tau)}^2\\
&= \frac{1}{T^2}\int_{t=0}^T\int_{\sigma=0}^T\cov_x(t-\sigma) \cov_y(\sigma-t)dt d\sigma \\
&=\frac{1}{T^2}\int_{t=0}^T\int_{\sigma=0}^T\int_\theta e^{i2\pi\theta(t-\sigma)}d\mu_x(\theta) \int_\phi e^{i2\pi\phi(\sigma-t)}d\mu_y(\phi)dt d\sigma\\
&=\frac{1}{T^2}\int_\theta\int_\phi\left(\int_{t=0}^T\int_{\sigma=0}^T e^{i2\pi(\theta-\phi)(t-\sigma)}dt d\sigma\right)d\mu_x(\theta)  d\mu_y(\phi)\\
&=\frac{1}{T}\int_\theta\int_\phi\left(\int_{t=-T}^T (1-|t|/T) e^{i2\pi(\theta-\phi)t}dt\right)d\mu_x(\theta)  d\mu_y(\phi)\\
&= \frac{1}{T}\int_\theta\int_\phi \kernel_T(\theta-\phi)d\mu_x(\theta)  d\mu_y(\phi).
\end{align*}
\end{proof}
%
%
%
\begin{proof}[Proof of Corollary \ref{cor:asymptote}]
We have that
\begin{align*}
	 \int_\theta\int_\phi &\kernel_T(\theta-\phi) d\mu_x(\theta)  d\mu_y(\phi) \\
	 =&  \int_\theta\int_\phi \kernel_T(\theta-\phi) \Phi_x(\theta) \Phi_y(\phi)d\theta d\phi \\
	 &\quad+ \sum_k \xamp_k^2 \int_\theta\int_\phi \kernel_T(\theta-\phi) \delta_{\freq_k^x}(\theta)\Phi_y(\phi) d\theta d\phi \\
	 &\quad+ \sum_\ell \yamp_\ell^2 \int_\theta\int_\phi \kernel_T(\theta-\phi) \delta_{\freq_\ell^y}(\phi)\Phi_x(\theta) d\theta d\phi \\
	 &\quad+ \sum_{k,\ell} \xamp_k^2 \yamp_\ell^2 \int_\theta\int_\phi \kernel_T(\theta-\phi) \delta_{\freq_k^x}(\theta)\delta_{\freq_\ell^y}(\phi) d\theta d\phi.
\end{align*}
First, note that as $\Phi_x, \Phi_y \in L_1$ and as $f_T$ is an approximate identity, it follows that $\kernel_T *\Phi_x \to \Phi_x$ and $\kernel_T *\Phi_y \to \Phi_y$ in $L_1$, as $T \to \infty$. Thus,
\begin{align*}
	\abs{\int_\theta\int_\phi \kernel_T(\theta-\phi) \Phi_x(\theta) \Phi_y(\phi)d\theta d\phi - \int_\theta \Phi_x(\theta) \Phi_y(\theta)d\theta} \to 0
\end{align*}
and as $\int_\theta  \kernel_T(\theta-\phi)\delta_{\freq_k^x}(\theta) d\theta = \kernel_T(\freq_k^x-\phi)$,
\begin{align*}
	\abs{\int_\theta\int_\phi \kernel_T(\theta-\phi) \delta_{\freq_k^x}(\theta)\Phi_y(\phi) d\theta d\phi - \Phi_y(\freq_k^x)} \to 0
\end{align*}
as $T \to \infty$. Finally, $\int_\theta\int_\phi \kernel_T(\theta-\phi) \delta_{\freq_k^x}(\theta)\delta_{\freq_\ell^y}(\phi) d\theta d\phi =  \kernel_T(\freq_k^x-\freq_\ell^y)$ and
\begin{align*}
	\abs{\kernel_T(\freq) - T \chi_{\{\freq = 0\}}} \to 0
\end{align*}
pointwise as $T\to \infty$. The statement of the proposition follows directly.
\end{proof}
%
%
%
%
%
\begin{proof}[Proof of Theorem \ref{thm:conv_in_dist}]
By Lemma~\ref{lemma:covariance_convergence}, it holds that the covariance function of $x^{(n)}$ converges to the covariance function of $x$. Thus, in order to prove the theorem, it is sufficient to show that $x^{(n)}$ converges in distribution to a Gaussian process.
Let $\boldsymbol{\tau} \in \RR^N$, for $N\in \RN$, be a set of sampling times, and let $X_k$ be the random vector defined as $X^{(n)}_k = \left[\begin{array}{ccc} x_k^{(n)}(\tau_1) & \ldots & x_k^{(n)}(\tau_N)\end{array}\right]^T$, where $x_k^{(n)}(t) = \sqrt{ \frac{B}{n}\Phi(\freq_k^{(n)})} e^{i2\pi\freq^{(n)}_k t}\randamp_k^{(n)}$. Furthermore, let $S_n = \sum_{k=1}^nX_k^{(n)}$. Then, as the vectors $X_k^{(n)}$ are independent, the covariance matrix of $S_n$ is given by
\begin{align*}
	\expop(S_nS_n^H) = \frac{B}{n} C_n \;,\; C_n = \sum_{k=1}^n \Phi(\freq_k^{(n)}) \fouriervec(\freq_k^{(n)})\fouriervec(\freq_k^{(n)})^H,
\end{align*}
where $\fouriervec: [-\pi,\pi) \to \RC^N$ is the sub-sampled Fourier vector corresponding to the sampling times $\boldsymbol{\tau}$. Then,
\begin{align*}
	\expop(S_nS_n^H)^{-1/2}X_k^{(n)} = \randamp_k^{(n)}\sqrt{\Phi(\freq_k^{(n)})}C_ n^{-1/2} \fouriervec(\freq_k^{(n)})
\end{align*}
and 
\begin{align*}
	\norm{\expop(S_nS_n^H)^{-1/2}X_k^{(n)}}_2^2 = \abs{\randamp_k^{(n)}}^2\Phi(\freq_k^{(n)})\fouriervec(\freq_k^{(n)}) C_ n^{-1} \fouriervec(\freq_k^{(n)}).
\end{align*}
As $\Phi$ is of bounded variation, and as $\fouriervec$ is a continuous function defined on a compact set, for $n > M$ for some finite $M$,
\begin{align*}
	C_{2n} \approx 2 C_n.
\end{align*}
Thus, for large $n$,
\begin{align*}
	\fouriervec(\freq_k^{(n)}) C_ n^{-1} a(\freq_k^{(n)}) \leq c/n
\end{align*}
where $c$ is a constant not depending on $n$ or $k$. Then,
\begin{align*}
	&\sum_{k=1}^n \expect{\norm{\expop(S_nS_n^H)^{-1/2}X_k^{(n)}}_2^3} \\
	&\leq  \sum_{k=1}^n \expect{|\randamp_k^{(n)}|^3}\Phi(\freq_k^{(n)})^{3/2} (c/n)^{3/2} \\
	&\leq \mu_3\frac{c^{3/2}}{\sqrt{n}} \max_\freq \Phi(\freq)^{3/2}\to 0,
\end{align*}
when $n\to \infty$, as $\mu_3 \triangleq \expect{|\randamp_k^{(n)}|^3}$ is finite by assumption of a finite fourth moment, and as $\Phi$ is bounded. According to the Lyapunov-type central limit theorem \cite{Bentkus05_49,Raic18_arxiv}, $S_n$ then converges in distribution to a Gaussian distribution as $n \to \infty$. This holds for any finite sample length $N$, with the requirement $N < n$ for invertibility of $C_n$. The statement of the theorem then follows directly.
\end{proof}
%
\begin{lemma} \label{lemma:covariance_convergence}
Let $\Phi$ be a continuous spectrum $\Phi\in C(\freqband)$ with support $\freqband = [\freq_c - B/2,\freq_c+B/2]$, where $\freq_c$ is the center frequency and $B$ is the bandwidth. Consider the sequence of stochastic processes
\begin{align*}
	x^{(n)}(t) = \sum_{k=1}^n \sqrt{ \frac{B}{n}\Phi(\freq_k^{(n)})} e^{i2\pi\freq^{(n)}_k t} \randamp_k^{(n)},
\end{align*}
where $\freq_k^{(n)}$ defines a uniform grid on $\mathcal{I}_B$, and where $\randamp_k^{(n)}$ are independent zero-mean stochastic variables such that $\expect{\abs{\randamp_k^{(n)}}^2} = 1$. Then, as $n\to \infty$, the covariance function of $x^{(n)}$ converges to the covariance function defined by $\Phi$.
\end{lemma}
\begin{proof}
We have
\begin{align*}
	\expop\left(x^{(n)}(t) \overline{x^{(n)}(t-\tau)}\right) = \sum_{k=1}^n \frac{B}{n}\Phi(\freq_k^{(n)})e^{i2\pi\freq^{(n)}_k \tau}.
\end{align*}
Then, as $\Phi$ is continuous on a compact interval, the Riemann sum on the right-hand side converges point-wise, i.e., for every lag $\tau$,
\begin{align*}
	\sum_{k=1}^n \frac{B}{n}\Phi(\freq_k^{(n)})e^{i2\pi\freq^{(n)}_k \tau} \to \int_{\mathcal{I}_B} \Phi(\freq)e^{i2\pi\freq\tau}d\freq,
\end{align*}
which is the covariance function associated with $\Phi$.
\end{proof}
%
%
%
\begin{proof}[Proof of Theorem~\ref{thm:approximation_variance}]
For the approximations in Theorem~\ref{thm:conv_in_dist}, the approximating spectra are of the form
\begin{align*}
	d\mu^{(n)}(\freq) &= \sum_{k=1}^n  \frac{B}{n}\breve{\Phi}\left(\frac{B}{n}k\right) \delta\left(\freq-\frac{B}{n}k\right)\\
	&=\sum_{k=1}^n  \frac{B}{n}\breve{\Phi}\left(\omega\right) \delta\left(\freq-\frac{B}{n}k\right)
\end{align*}
where, for notational convenience, $\breve{\Phi}(\theta) \triangleq \Phi(\theta+\freq_c-B/2)$.
Then,
\begin{align*}
	 \int_\phi\!\! f_T(\theta\!-\!\phi)d\mu^{(n)}\!(\phi) \!&= \int_\phi f_T(\phi)d\mu^{(n)}(\theta-\phi) \\
	 &=\! \frac{B}{n} \sum_{m=-\infty}^\infty \kernel_T\left(\frac{Bm}{n}\right) \breve{\Phi}\left(\theta-\frac{Bm}{n}\right) \\
	 &= \!\!\int_{\phi} \!\sum_{m=-\infty}^\infty\!\!\! \frac{B}{n}\!\kernel_T\!\left(\phi\right)\!\delta\!\left(\!\!\phi\!-\!\frac{Bm}{n}\!\right)\! \breve{\Phi}(\theta\!-\!\phi)d\phi \\
	 &=\!\! \int_\phi \kernel_T^{(n)}(\phi)\breve{\Phi}(\theta\!-\!\phi)d\phi,
\end{align*}
where the summation limits in the second equality follows as the support of $\Phi$ is limited to $\freqband$, and where
\begin{align*}
	\kernel_T^{(n)}(\phi) \triangleq \sum_{m=-\infty}^\infty \kernel_T(\phi) \delta\left(\phi - \frac{Bm}{n}\right)\frac{B}{n}.
\end{align*}
According to Lemma~\ref{lemma:summed_kernel}, $\frac{1}{\asympratio+\rho(\asympratio)}\kernel_T^{(n)}$ acts as an approximate identity as $n \to \infty$. Thus, for fixed $\asympratio$,
\begin{align*}
	 \int_\phi f_T(\theta-\phi)d\mu^{(n)}(\phi) \to (\asympratio+\rho(\asympratio)) \breve{\Phi}(\theta).
\end{align*}
Then, as
\begin{align*}
	 \int_{\freqband}\breve{\Phi}(\theta)d\mu^{(n)}(\theta) &= \frac{B}{n} \sum_{k=1}^n\left(\breve{\Phi}\left(\frac{Bk}{n}\right)\right)^2 \to \int_{\freqband} \breve{\Phi}(\theta)^2d\theta,
\end{align*}
we have
\begin{align*}
	 \int\!\!\!\int f_T(\theta-\phi)d\mu^{(n)}(\phi)d\mu^{(n)}(\theta) \to (\asympratio+\rho(\asympratio)) \int_{\freqband} \Phi(\theta)^2d\theta.
\end{align*}
The statement of the theorem then follows directly from Theorem~\ref{thm:auto_cov_general}.
\end{proof}
%
%
%
\begin{lemma}\label{lemma:summed_kernel}
Let $\asympratio = \frac{BT}{n}$ be fixed. Then, the function
\begin{align*}
	\frac{1}{\asympratio+\rho(\asympratio)} \kernel_T^{(n)}(\phi) = \frac{1}{\asympratio+\rho(\asympratio)}\sum_{m=-\infty}^\infty \kernel_T(\phi) \delta\left(\phi - \frac{Bm}{n}\right)\frac{B}{n},
\end{align*}
parametrized by $n$, is an approximate identity, i.e., for any $\Phi \in L_1$, $\frac{1}{\asympratio+\rho(\asympratio)} \kernel_T^{(n)} * \Phi \to \Phi$ in $L_1$ as $n \to \infty$.
\end{lemma}
\begin{proof}
Firstly, note that for any $T> 0$, $\kernel_T(\theta) = T \kernel_1(T\theta)$. We have
\begin{align*}
	\sum_{m=-\infty}^\infty\!\kernel_T(\phi) \delta\!\left(\phi \!-\! \frac{Bm}{n}\right)\!\frac{B}{n}\!&=\!\sum_{m=-\infty}^\infty\!\kernel_{\frac{BT}{n}}\!\left(\frac{n}{B}\phi\right)\!\delta\!\left(\!\phi\!-\!\frac{Bm}{n}\!\right) \\
	&=\!\sum_{m=-\infty}^\infty\!\kernel_{\asympratio}\!\left(\frac{n}{B}\phi\right)\!\delta\left(\!\phi\!-\!\frac{Bm}{n}\!\right)\!.
\end{align*}
Then,
\begin{align*}
	\int_{-\infty}^\infty \kernel_T^{(n)}(\phi) d\phi &= \sum_{m=-\infty}^\infty \kernel_\asympratio(m) = \kernel_\asympratio(0) + 2\sum_{m=1}^\infty \kernel_\asympratio(m) \\
	&= \asympratio + \frac{4}{\asympratio (2\pi)^2}\sum_{m=1}^\infty \frac{1-\cos(2\pi \asympratio m)}{m^2}.
\end{align*}
Clearly, for $\asympratio \in \RN$, all terms in the series are zero. For non-integer $\asympratio$, let $\breve{\asympratio} = \asympratio -  \lfloor\asympratio\rfloor$, with $\lfloor\asympratio\rfloor$ denoting the integer part of $\asympratio$. Then,
\begin{align*}
	&\int_{-\infty}^\infty \kernel_T^{(n)}(\phi) d\phi \\
	&= \asympratio + \frac{1}{\asympratio \pi^2}\left( \frac{\pi^2}{6} - \frac{1}{2}\left( \mathrm{Li}_{2}(e^{2i\pi\asympratio}) + \mathrm{Li}_{2}(e^{-2i\pi\asympratio})\right) \right) \\
	&= \asympratio + \frac{1}{\asympratio \pi^2}\left( \frac{\pi^2}{6} - \frac{1}{2}\left( \mathrm{Li}_{2}(e^{2i\pi\breve{\asympratio}}) + \mathrm{Li}_{2}(e^{-2i\pi\breve{\asympratio}})\right) \right) \\
	&= \asympratio + \frac{1}{\asympratio \pi^2}\left( \frac{\pi^2}{6} - \frac{1}{2}\left( -\frac{(2i\pi)^2}{2!}B_2(\breve{\asympratio})\right) \right) \\
	&= \asympratio +\frac{ \breve{\asympratio}- \breve{\asympratio}^2}{\asympratio} = \asympratio+\rho(\asympratio),
\end{align*}
where $\mathrm{Li}_{2}$ is the polylogarithm and $B_2$ is the Bernoulli polynomial $B_2(x) = x^2-x+\frac{1}{6}$. Furthermore, letting $n\to \infty$, and thereby also $T\to\infty$ as $\asympratio$ is fixed, $\kernel_{\asympratio}(\frac{n}{B}\phi) = \frac{B}{n}\kernel_{T}(\phi) \to 0$ for $\abs{\phi}>0$, implying
\begin{align*}
	\int_{\phi \notin [-\epsilon,\epsilon]} \kernel_T^{(n)}(\phi) d\phi \to 0
\end{align*}
for any $\epsilon > 0$. The statement of the lemma follows.
\end{proof}
\balance
\bibliographystyle{plain}
\bibliography{ElvanderK21_SP_GENERAL_spformat_ARXIV.bbl}
%
%

%
%
\end{document}